\documentclass[a4paper,11pt]{article}
\usepackage[margin=0.95in]{geometry}
\usepackage[numbers,sort]{natbib}
\usepackage[T1]{fontenc}
\usepackage{authblk}
\usepackage{amsmath,amsthm,amssymb,thm-restate}
\usepackage{algorithm}
\usepackage[noend]{algpseudocode}
\usepackage{dsfont}
\usepackage{tikz}
\usepackage{comment}
\usepackage{bm}

\usepackage[framemethod=tikz]{mdframed}

\usepackage{cleveref}
\usepackage{xspace}
\usepackage{xcolor}
\usepackage[colorinlistoftodos,textsize=tiny,textwidth=2cm,color=red!25!white,obeyFinal]{todonotes}

\theoremstyle{plain}
\newtheorem{remark}{Remark}
\newtheorem{thm}{Theorem}[section]
\newtheorem{cor}[thm]{Corollary}
\newtheorem{prop}[thm]{Proposition}
\newtheorem{conj}[thm]{Conjecture}

\newtheorem{lem}[thm]{Lemma}
\newtheorem{Def}[thm]{Definition}
\newtheorem{obs}[thm]{Observation}

\newcommand{\E}{\mathbb{E}\xspace}

\newcommand{\Cov}{\mathrm{Cov}\xspace}

\newcommand{\eat}[1]{}

\newcommand{\eps}{\epsilon}

\newcommand{\halfplusbestc}{0.527}
\newcommand{\bestc}{0.027}
\newcommand{\approxratio}{1.897}

\newcommand{\calM}{\mathcal{M}}
\newcommand{\calA}{\mathcal{A}}

\newenvironment{wrapper}[1]
{
	\smallskip
	\begin{center}
		\begin{minipage}{\linewidth}
			\begin{mdframed}[hidealllines=true, backgroundcolor=gray!20, leftmargin=0cm,innerleftmargin=0.4cm,innerrightmargin=0.4cm,innertopmargin=0.4cm,innerbottommargin=0.4cm,roundcorner=10pt]
				#1}
			{\end{mdframed}
		\end{minipage}
	\end{center}
	\smallskip
}

\usepackage{etoolbox}
\usepackage{tikz}
\usetikzlibrary{tikzmark}
\usetikzlibrary{calc}

\errorcontextlines\maxdimen

\newcommand{\ALGtikzmarkcolor}{black}
\newcommand{\ALGtikzmarkextraindent}{4pt}
\newcommand{\ALGtikzmarkverticaloffsetstart}{-.5ex}
\newcommand{\ALGtikzmarkverticaloffsetend}{-.5ex}
\makeatletter
\newcounter{ALG@tikzmark@tempcnta}
\newcommand\ALG@tikzmark@start{%
	\global\let\ALG@tikzmark@last\ALG@tikzmark@starttext%
	\expandafter\edef\csname ALG@tikzmark@\theALG@nested\endcsname{\theALG@tikzmark@tempcnta}%
	\tikzmark{ALG@tikzmark@start@\csname ALG@tikzmark@\theALG@nested\endcsname}%
	\addtocounter{ALG@tikzmark@tempcnta}{1}%
}

\def\ALG@tikzmark@starttext{start}
\newcommand\ALG@tikzmark@end{%
	\ifx\ALG@tikzmark@last\ALG@tikzmark@starttext
	\else
	\tikzmark{ALG@tikzmark@end@\csname ALG@tikzmark@\theALG@nested\endcsname}%
	\tikz[overlay,remember picture] \draw[\ALGtikzmarkcolor] let \p{S}=($(pic cs:ALG@tikzmark@start@\csname ALG@tikzmark@\theALG@nested\endcsname)+(\ALGtikzmarkextraindent,\ALGtikzmarkverticaloffsetstart)$), \p{E}=($(pic cs:ALG@tikzmark@end@\csname ALG@tikzmark@\theALG@nested\endcsname)+(\ALGtikzmarkextraindent,\ALGtikzmarkverticaloffsetend)$) in (\x{S},\y{S})--(\x{S},\y{E});%
	\fi
	\gdef\ALG@tikzmark@last{end}%
}

\apptocmd{\ALG@beginblock}{\ALG@tikzmark@start}{}{\errmessage{failed to patch}}
\pretocmd{\ALG@endblock}{\ALG@tikzmark@end}{}{\errmessage{failed to patch}}
\makeatother

\algblock[with]{With}{EndWith}
\algblockdefx[With]{With}{EndWith}%
[1]{\textbf{with} #1 \textbf{do}}%
{}
\makeatletter
\ifthenelse{\equal{\ALG@noend}{t}}%
{\algtext*{EndWith}}
{}%
\makeatother

\title{The Greedy Algorithm is \emph{not} Optimal for On-Line Edge Coloring}

\author{Amin Saberi}
\author{David Wajc}
\affil{Stanford University\\
	\{saberi, wajc\}@stanford.edu}
\date{\vspace{-1cm}}

\begin{document}

\maketitle


\begin{abstract}
	
Nearly three decades ago, Bar-Noy, Motwani and Naor showed that no online edge-coloring algorithm can edge color a graph optimally.
Indeed, their work, titled ``the greedy algorithm is optimal for on-line edge coloring'', shows that the competitive ratio of $2$ of the na\"ive greedy algorithm is best possible online.
However, their lower bound required bounded-degree graphs, of maximum degree $\Delta = O(\log n)$, which prompted them to conjecture that better bounds are possible for higher-degree graphs.
While progress has been made towards resolving this conjecture for restricted inputs and arrivals or for random arrival orders, an answer for fully general \emph{adversarial} arrivals remained elusive.

%
We resolve this thirty-year-old conjecture in the affirmative, presenting a $(1.9+o(1))$-competitive online edge coloring algorithm for general graphs of degree $\Delta = \omega(\log n)$ under vertex arrivals.
At the core of our results, and of possible independent interest, is a new online algorithm which rounds a fractional bipartite matching $x$ online under vertex arrivals, guaranteeing that each edge $e$ is matched with probability $(1/2+c)\cdot x_e$, for a constant $c>\bestc$.

%
%
\end{abstract}


\section{Introduction}

An edge coloring of a graph is a decomposition of its edge-set into few vertex-disjoint edge-sets (matchings), or \emph{colors}. 
Edge coloring a graph of maximum degree $\Delta$ trivially requires at least $\Delta$ colors, and this is tight for bipartite graphs, by the century-old result of K\"onig \cite{konig1916graphen}. For general graphs, $\Delta$ colors are not always sufficient (e.g., in odd-length cycles), yet $\Delta+1$ colors are always sufficient, by Vizing's Theorem \cite{vizing1964estimate}. 

Algorithmically matching, or approximating, the optimal $\Delta(+1)$ colors needed to edge color a graph has been the focus of much concentrated effort, for numerous computational models. 
These include offline, online, distributed, parallel, and dynamic algorithms (see, e.g.,  \cite{cole2001edge,cohen2019tight,su2019towards,chang2018complexity,karloff1987efficient,motwani1994probabilistic,duan2019dynamic,charikar2021improved,wajc2020rounding} and references therein).
These different models' specific challenges naturally impose limitations on the attainable approximations.
For example, Holyer's Theorem \cite{holyer1981np} rules out efficient offline algorithms for computing an optimal edge coloring in general graphs, unless \textsc{P}=\textsc{NP}.

For online algorithms, the challenge is in making immediate and irrevocable decisions concerning edges' colors after only part of the input is revealed. 
For example, the input graph can either be revealed edge-by-edge (edge arrivals) or vertex-by-vertex (vertex arrivals), and an online algorithm must assign colors to edges after they are revealed, immediately and irrevocably.
The measure of an online algorithm is its competitive ratio, which is the worst-case ratio of the number of colors used by the algorithm to those of the optimal offline algorithm, namely, $\Delta$ or $\Delta+1$. 

In both the edge-arrival and vertex-arrival settings, a simple greedy algorithm has competitive ratio $2$. The natural question, then, is whether a better online algorithm exists. 
Some thirty years ago, Bar-Noy, Motwani and Naor~\cite{bar1992greedy} showed that this competitive ratio of $2$ is best possible, and no online algorithm (randomized or deterministic) can do better, in either arrival model.

However, noting that their result only holds for bounded-degree $n$-node graphs, of maximum degree $\Delta = O(\log n)$, Bar-Noy et al.~conjectured that better algorithms exist for graphs of sufficiently high maximum degree.

\begin{wrapper}
\begin{conj}[\cite{bar1992greedy}]\label{conj:bmn}
	There exists a $(2-\Omega(1))$-competitive online edge coloring algorithm under vertex arrivals in $n$-node graphs of maximum degree $\Delta =\omega(\log n)$.
\end{conj}
\end{wrapper}

Bar-Noy et al.~conjectured that the same holds under the more challenging edge-arrival model, and that moreover a $(1+o(1))$-competitive algorithm exists. 
These conjectures remain out of reach, though progress has been made on them over the years. 
For edge arrivals, a positive resolution of the stronger conjecture was achieved under the assumption of \emph{random order} arrivals, where the input is generated adversarially, but its arrival order is randomly permuted by nature \cite{aggarwal2003switch,bahmani2012online,bhattacharya2021online}.
For adversarial vertex arrivals, Cohen et al.~\cite{cohen2019tight} showed that for \emph{bipartite} graphs under one-sided vertex arrivals (vertices of one side are given, and the other side's vertices arrive), the conjectured $(1+o(1))$-competitive ratio is achievable for $\Delta=\omega(\log n)$.
Whether the competitive ratio of $2$ of the greedy algorithm is optimal under \emph{general} vertex arrivals, in \emph{general} graphs, however, remained open.

We answer the above open question, resolving \Cref{conj:bmn} in the affirmative.

\begin{wrapper}
\begin{restatable}{thm}{greedysubopt}\label{thm:greedy-suboptimal}
	There exists an online edge coloring algorithm which is $(\approxratio+o(1))$-competitive w.h.p.~on general $n$-node graphs with maximum degree $\Delta = \omega(\log n)$ under vertex arrivals.
\end{restatable}
\end{wrapper}

\begin{remark}\label{remark1}
	For general $\Delta$, the $o(1)$ term in the above theorem is of the form $\sqrt[\gamma]{\log n /\Delta}$, for some constant $\gamma > 0$. This implies a better than two approximation ratio for sufficiently large $\Delta=O(\log n)$. For simplicity of exposition, we do not elaborate on this point.
\end{remark}

\subsection{Techniques}
To obtain our results, we combine and extend several previous algorithmic ideas. 

Our starting point is the following natural recursive approach, due to Karloff and Shmoys \cite{karloff1987efficient}, which reduces edge coloring a general graph $G$ to edge coloring random bipartite subgraphs. 
Their idea was to assign each vertex to either side of a random subgraph uniformly, resulting in a bipartite subgraph $H$ of $G$ with maximum degree $\Delta/2+o(\Delta)$ for $\Delta=\omega(\log n)$, by standard tail bounds. 
Consequently, applying an $\alpha$-approximate algorithm to the random bipartite graph and recursing on the remaining edges is easily shown to result in an edge coloring using $\alpha\cdot \Delta/2 + o(\Delta) + \alpha\cdot \Delta/4 + o(\Delta) \dots = \alpha\cdot \Delta + o(\Delta)$ colors.
Importantly for us, this approach, originally used in the context of NC algorithms by \cite{karloff1987efficient}, is implementable online, by sampling the random bipartitions in advance. (See \Cref{sec:karloff-shmoys}.)

At this point, one might be tempted to use the online algorithm of Cohen et al.~\cite{cohen2019tight} for these random bipartite subgraphs. Unfortunately, the reduction of Karloff and Shmoys \cite{karloff1987efficient} applied to online edge coloring with general vertex arrivals requires an online algorithm for bipartite graphs with \emph{interleaved} arrivals, and not one-sided arrivals, as handled by \cite{cohen2019tight}. To instantiate the Karloff-Shmoys approach, we therefore present a $(2-c)$-competitive edge coloring algorithm for interleaved vertex arrivals in bipartite graphs, which, when combined with the approach of \cite{karloff1987efficient}, then extends to general graphs.

To obtain an edge-coloring algorithm for bipartite graphs under interleaved vertex arrival, we extend the approach of Cohen et al.~\cite{cohen2019tight}, who showed that an $(\alpha+o(1))$-competitive edge coloring can be achieved by repeatedly applying a matching algorithm which matches each edge with probability $(1/\alpha)/\Delta$. 
For each vertex of degree $\Delta(1-o(1))$, such a matching results in $v$ being matched with probability $(1/\alpha)\cdot (1-o(1))$. Repeating the above a super-logarithmic number of times (making use of $\Delta=\omega(\log n)$) therefore decreases the maximum degree of the graph at a rate of roughly one per $\alpha$ colors used.
Cohen et al.~used this approach with $\alpha=1+o(1)$, using an online matching algorithm from \cite{cohen2018randomized}, on bipartite graphs under one-sided arrivals. 
We observe that this approach extends to arbitrary $\alpha$ and any arrival model, including interleaved vertex arrivals in bipartite graphs. (See \Cref{sec:reducing-coloring-to-matching}.)

Motivated by the above discussion, we design an online matching algorithm for bipartite graphs under interleaved arrivals, which matches each edge with probability $(1/2+c)/\Delta$, for some constant $c>0$. 
More generally, and of possible independent interest, we design an online rounding algorithm for bipartite fractional matchings under interleaved vertex arrivals, with a multiplicative factor of $1/2+c$.
That is, we show how, given a bipartite graph $G$ and a fractional matching $x$ in $G$ revealed vertex-by-vertex, one can output a randomized matching which matches each edge $e$ in $G$ with probability $(1/2+c)\cdot x_e$. 
This extends a similar online rounding algorithm previously developed by the authors with Papadimitriou and Pollner \cite{papadimitriou2021online} in the context of online stochastic optimization, but which only works under one-sided vertex arrivals, and is therefore insufficient for our needs.
This new rounding algorithm is the technical meat of this paper, and is presented in \Cref{sec:rounding}.

Combining the above, we obtain \Cref{thm:greedy-suboptimal}, and the positive resolution of \Cref{conj:bmn}.

\subsection{Related Work}\label{sec:related-work}
The first positive results for online edge coloring were under random order edge arrivals.
In this setting, Aggarwal et al.~\cite{aggarwal2003switch} showed that a $(1+o(1))$-competitive ratio is achievable in dense multigraphs with maximum degree $\Delta=\omega(n^2)$. Bahmani et al.~\cite{bahmani2012online} then showed that the greedy algorithm is sub-optimal for any graph of  maximum degree $\Delta= \omega(\log n)$. Achieving the best of both these results, Bhattacharya et al.~\cite{bhattacharya2021online} recently obtained a $(1+o(1))$-competitive algorithm for graphs of maximum degree $\Delta=\omega(\log n)$.
As stated above, the only prior algorithm which outperforms the greedy algorithm under \emph{adversarial} arrivals is the algorithm of Cohen et al.~\cite{cohen2019tight} for bipartite graphs under one-sided vertex arrivals. 
In this work, we remove the assumption of bipartiteness and one-sided arrivals, and show how to outperform greedy in general graphs under arbitrary vertex arrivals.

Our work also ties into the long line of work on online matching, initiated by Karp, Vaizrani and Vazirani \cite{karp1990optimal}. (See e.g., \cite{feng2021two,huang2020fully,gamlath2019online,fahrbach2020edge,naor2018near,ashlagi2019edge} and references therein and \cite{mehta2013online} for a survey of earlier work.) 
Historically, most research on online matching considered bipartite graphs with one-sided arrivals, due to applications in Internet advertising \cite{mehta2007adwords,feldman2009online2}.
A recent line of work considers such problems subject to interleaved vertex arrivals (motivated by more dynamic two-sided markets), as well as vertex arrivals in general graphs \cite{huang2020fully,huang2020fully2,gamlath2019online,ashlagi2019edge,wang2015two}. 
Our rounding algorithm for bipartite graphs with interleaved arrivals adds to the list of tools for tackling problems in this space.

Few of the works in the online (bipartite) matching literature rely on randomized rounding. At first blush, this seems surprising, given the integrality of the bipartite fractional matching polytope, and the multitude of competitive fractional algorithms for problems in this area \cite{kalyanasundaram2000optimal,feldman2009online2,huang2020fully,huang2020fully2,buchbinder2007online,wang2015two}.
However, as pointed out in \cite{devanur2013randomized} and elaborated upon in \cite{cohen2018randomized}, lossless rounding of a fractional matching $x$ is impossible in online settings. In particular, outputting a matching $\calM$ which matches each edge $e$ in a bipartite graph with probability $\Pr[e\in \calM]=x_e$ is impossible in online settings, though it is easy to do offline. A natural question, then, is what is the highest value of $\alpha<1$ for which one can guarantee $\Pr[e\in \calM]\geq \alpha\cdot x_e$ when rounding bipartite fractional matchings online.
The batched OCRS of Ezra et al.~\cite{ezra2020online} gives $\alpha=1/2$, unfortunately too low for our purposes. 
In prior work  \cite{papadimitriou2021online}, motivated by a variation of the online Bayesian selection problem, we improve this bound to $\alpha=0.51$, though only for one-sided arrivals, which is insufficient for our needs here.
In this work we generalize this result, achieving a slightly higher $\alpha=\halfplusbestc$, subject to \emph{interleaved} vertex arrivals.
\section{Preliminaries}\label{sec:prelims}

The underlying (a priori unknown) input to our problem is an $n$-node graph $G=(V,E)$ of maximum degree $\Delta$ (with $n$ and $\Delta$ both known). The vertices of $G$ are revealed over time. For notational convenience, we associate the $n:=|V|$ vertices with the numbers in $[n]$ by order of appearance, and denote by $u<v$ the fact that $u$ arrives before $v$. 
When a vertex $v$ arrives (at \emph{time} $v$), all its edges $(u,v)$ to its previously-arrived neighbors $u<v$ are revealed. 
After $v$ arrives, and before arrival of vertex $v+1$, an online edge coloring algorithm must decide, irrevocably, which color to assign to all edges $(u,v)$ with $u<v$.
The objective is to minimize the number of distinct colors used.

As outlined in the introduction, we will rely on the ability to edge color general graphs by recursively coloring random bipartite subgraphs, as first proposed by Karloff and Shmoys \cite{karloff1987efficient}, in the context of NC algorithms.
The extension and proof for online settings is essentially the same, and is provided, for completeness, in \Cref{sec:karloff-shmoys}.

\begin{restatable}{lem}{karloffshmoys}(Implied by \cite{karloff1987efficient})\label{random-subgraphs}
	Given an online edge coloring algorithm which is $\alpha$-competitive w.h.p.~on bipartite graphs of maximum degree $\Delta = \omega(\log n)$ under interleaved vertex arrivals, there exists an online edge coloring algorithm which is $(\alpha+o(1))$-competitive w.h.p.~on \emph{general} graphs of maximum degree $\Delta=\omega(\log n)$ under vertex arrivals.
\end{restatable}

The following lemma, implied by the recent work of Cohen et al.~\cite{cohen2019tight}, reduces $\alpha$-competitive edge coloring to online matching algorithms which match each edge with probability $(1/\alpha)/\Delta$. The proof is is provided, for completeness, in \Cref{sec:reducing-coloring-to-matching}.

\begin{restatable}{lem}{coloringtomatching}\label{coloring-to-matching}(Implied by \cite{cohen2019tight})
	Let $\calA$ be an online matching algorithm which on any (bipartite) graph of maximum degree $\Delta\leq \Delta'$ under vertex arrivals, matches each edge with probability at least $1/(\alpha\Delta')$. Then, there exists an online edge coloring algorithm $\calA'$ which is $(\alpha+o(1))$-competitive w.h.p.~for (bipartite) graphs of maximum degree $\Delta = \omega(\log n)$ under vertex arrivals.
\end{restatable}

Motivated by \Cref{coloring-to-matching}, we show how to (approximately) round fractional matchings online. 
These are assignments of nonnegative $x_e\geq 0$ to edges $e\in E$, satisfying the fractional matching constraint, $\sum_{e\ni v} x_e\leq 1$ for all $v\in V$. 
This is a fractional relaxation of the matching constraint, which stipulates that the degree of any vertex in a matching be at most one.
Fittingly, we refer to $\sum_{w<v} x_{u,w}$ as the \emph{fractional degree} of $u$ before arrival of $v$ (or at its arrival time, if $u=v$).
We shall show how to round fractional matchings up to a multiplicative error of $\alpha<2$. This rounding subroutine applied to the fractional matching assigning value $1/\Delta$ to each edge of the graph thus matches each edge with probability $1/(\alpha\Delta)$. Combined with lemmas \ref{random-subgraphs} and \ref{coloring-to-matching}, this yields our $(\alpha+o(1))\Delta$ coloring algorithm.

\subsection{Negative Association}
\label{sec:prelimNA} 

In our work we will need to bound positive correlations between variables. At the core of these proofs will be a use of \emph{negatively associated} random variables. 
This section introduces this notion of negative dependence and its properties which we use.

\begin{Def}[\cite{khursheed1981positive,joag1983negative}]\label{def:NA}
	Random variables $X_1,\dots,X_n$ are \emph{negatively associated (NA)} if every two monotone nondecreasing functions $f$ and $g$ defined on disjoint subsets of the variables in $\vec{X}$ are negatively correlated. That is,
	\begin{equation}\label{eq:NA}
	\E[f\cdot g] \leq \E[f]\cdot \E[g].
	\end{equation}
\end{Def}

The following simple example of NA variables will prove useful for us.

\begin{prop}[0-1 Principle \cite{dubhashi1996balls}]\label{0-1-NA}
	Let $X_1,\dots,X_n\in \{0,1\}$ be binary random variables satisfying $\sum_i X_i\leq 1$ always. Then, the variables $X_1,\dots,X_n$ are NA.
\end{prop}

Negative association is closed under several operations, allowing to construct more elaborate NA distributions from simpler NA distributions as above (see \cite{khursheed1981positive,joag1983negative,dubhashi1996balls}).
\begin{prop}[Independent Union]\label{NA:ind-union}
	Let $X_1,\dots,X_n$ be NA and $Y_1,\dots,Y_m$ be NA, with $\{X_i\}_i$ independent of $\{Y_j\}_j$. Then, the variables $X_1,\dots,X_n,Y_1,\dots,Y_m$ are all NA.
\end{prop}
\begin{prop}[Function Composition]\label{NA:fn-comp}
	Let $X_1,\dots,X_n$ be NA variables, and let $f_1,\dots,f_k$ be monotone nondecreasing functions defined on disjoint subsets of the variables in $\vec{X}$. Then the variables $f_1(\vec{X}),\dots,f_k(\vec{X})$ are NA.
\end{prop}

An immediate corollary of negative association, obtained by considering the functions $f(\vec{X})=X_i$ and $g(\vec{X})=X_j$ for $i\neq j$, is pairwise negative correlation.

\begin{prop}[NA implies Negative Correlation]\label{NA:neg-corr}
	Let $X_1,\dots,X_n$ be NA variables. Then, for all $i\neq j$, we have that
	$\Cov(X_i,X_j)\leq 0$. 
\end{prop}
\section{Rounding Bipartite Fractional Matchings Online}\label{sec:rounding}

In this section we present an online algorithm which (approximately) rounds a bipartite fractional matching under interleaved vertex arrivals. 
In what follows, we let $c\geq \bestc$ be the largest value below $0.03$ satisfying
\begin{align}\label{def:c}
	(1/2 - c)(1-4c)(1/2 - c - 6c/(1/2-c)) - 2c & \geq 0.
\end{align}
We note that this choice of $c\leq 0.03$ also satisfies the following.\footnote{We encourage the reader to think of $c\to 0$, and note that inequalities \eqref{def:c} and \eqref{properties:c} hold for sufficiently small constant $c>0$. Our choice of $c\approx \bestc$ is simply the largest satisfying all these constraints.}
\begin{align}\label{properties:c}
\min\{1/2-c,\, 1-4c,\, 1 - 6c/(1/2-c)^2\} \geq 0.
\end{align}
We show the following.

\begin{wrapper}
\begin{thm}\label{per-edge-guarantees}
	There exists an online algorithm which, given an (unknown) bipartite graph $G$ under interleaved vertex arrivals, together with a fractional matching $x$ in $G$, outputs a random matching $\calM$ matching each edge $e\in E$ with probability
	\begin{equation}\label{invariant}
	\Pr[e\in \calM] = (1/2+c)\cdot x_e \geq \halfplusbestc \cdot x_e.
	\end{equation}
\end{thm}
\end{wrapper}

We now turn to describing the algorithm claimed by the above theorem.

\subsection{Intuition and Algorithm}
Before presenting our algorithm, we describe the approach used to obtain \Cref{per-edge-guarantees} under one-sided arrivals \cite{papadimitriou2021online}, and then discuss the new ideas needed to extend this result to interleaved arrivals.

Naturally, an edge $(u,v)$ with $u<v$ (i.e., $v$ arriving later than $u$) can only be matched if $u$ is not already matched before the arrival of $v$. We denote by $F_{u,v}$ the event that $u$ is free (i.e., is not matched in $\calM$) prior to the arrival of $v$. 
The guarantee of \Cref{per-edge-guarantees} implies the following closed form for the probability of this event.
\begin{equation}\label{prob-free}
\Pr[F_{u,v}] = g(u,v) := 1 - \sum_{w < v} (1/2+c)\cdot x_{u,w}.
\end{equation}

To achieve marginal probabilities of $\Pr[(u,v)\in \calM] = (1/2+c)\cdot x_{u,v}$, our first step is to have every arriving vertex $v$ pick a random neighbor $u<v$ with probability $x_{u,v}$, and then, if $u$ is free, we match $(u,v)$ with probability $q_{u,v}:=\min(1,(1/2+c)/g(u,v))$. For neighbors $u$ of low fractional degree upon arrival of $v$, i.e., $\sum_{w<v} x_{u,w} \leq \frac{1/2-c}{1/2+c}$, this last probability is precisely $q_{u,v}=(1/2+c)/\Pr[F_{u,v}]$. Consequently, we match each such edge $(u,v)$ with probability $\Pr[(u,v)\in \calM] = x_{u,v}\cdot \Pr[F_{u,v}]\cdot (1/2+c)/\Pr[F_{u,v}] = (1/2+c)\cdot x_{u,v}$, as desired.
For edges $(u,v)$ for which $u$ has \emph{high} fractional degree, on the other hand, this only gives us $\Pr[(u,v)\in \calM] \geq (1/2-c)\cdot x_{u,v}$, and this can be tight.



To increase the probability of an edge $(u,v)$ to be matched to the desired $(1/2+c)\cdot x_{u,v}$, we repeat this process a second time, making a second pick, if $v$ is not matched after its first pick. 
Here, we must argue that the variables $\{F_{u,v}\mid u<v\}$ do not have strong positive correlation.
Indeed, if, as an extreme case, we had $F_{u,v} = F_{w,v}$ always for all $u,w<v$, and $v$ had only high-degree neighbors (for which $q_{u,v}=1$), then if $v$ is not matched to its first pick, then all its neighbors must be matched, and $v$ is therefore never matched as a second pick. 
This implies that a second pick does not increase $\Pr[(u,v)\in \calM]$ in this case.
As shown in \cite{papadimitriou2021online}, under one-sided arrivals, this problematic scenario does not occur, since the matched status of neighbors of $v$ is rather weak.
For interleaved arrivals, however, the underlying argument does not carry through, as we now explain.

\subsubsection{Extension to Interleaved Arrivals}

The key difference between one-sided and interleaved arrivals is that now we require small positive correlation between the matched statuses of every two nodes on the same side of the bipartition, rather than just nodes on the ``offline side''.
For one-sided arrivals, the weak positive correlation between offline vertices was due to two factors. 
(1) low-degree offline vertices are matched only due to semi-adaptive matching choices, where precisely one neighbor of an arriving online vertex is picked, and at most one is matched. (That is, they are only matched as a first pick.) Therefore, by the 0-1 Principle (\Cref{0-1-NA}) and closure properties of NA distributions (propositions \ref{NA:ind-union} and \ref{NA:fn-comp}), the indicators for a vertex to be matched when it has low fractional degree are NA, and hence are negatively correlated. 
(2) On the other hand, the probability of a node to be matched when it has high degree is low, since each edge is matched with probability $(1/2+c)\cdot x_{u,v}$, and the residual fractional degree when $v$ has high degree is $1-\frac{1/2-c}{1/2+c} = \frac{2c}{1/2+c} \leq 4c$.
Putting (1) and (2) together, we find that the matched statuses of any two offline vertices have small correlation.

Unfortunately, under interleaved arrivals, the above is no longer true. In particular, if a vertex $v$ has low fractional degree upon arrival, it may still be matched as a second pick upon arrival (due to its high-degree neighbors). Consequently, the indicators for vertices on the same side of the bipartition being matched when they have low fractional degree are no longer negatively associated, thus undoing the entire argument used to bound $\Cov(F_{u,v}, F_{w,v})$ for vertices $u,w<v$ on the same side of the bipartition.

To overcome this problem, we have each arriving vertex $v$ with low fractional degree upon arrival only pick once, and rely on its low fractional degree to pick each neighbor with higher probability.
In particular, when such a vertex $v$ arrives, we pick at most one neighbor with probability $x_{u,v}\cdot \frac{1/2+c}{1/2-c}$. (Since $v$ has low fractional degree on arrival, $\sum_{u<v} x_{u,v}\leq \frac{1/2-c}{1/2+c}$, this is well-defined.)
Then, if this picked vertex $u$ is free, we match $(u,v)$ with probability $\frac{1/2-c}{\Pr[F_{u,v}]} = \frac{1/2-c}{g(u,v)}(\leq 1)$, resulting in the edge $(u,v)$ being matched with probability $x_{u,v}\cdot (1/2+c)$. 
Crucially for our analysis, this now allows us to show that the indicators for vertices (in the same side of the graph) to be matched when they have low fractional degree is again negatively associated. This then results in the matched status of vertices again being decomposable into two variables, with the first being negatively correlated, and the second having low probability, from which we obtain that vertices on the same side of the bipartition have low correlation.\footnote{We note that  Gamlath et al.~\cite{gamlath2019online} followed a superficially similar rounding approach, using two choices. 
As they only required bounds on the (unweighted) matching's size, their analysis relied on showing that \emph{globally} positive correlation is low. As we desire high matching probability on an edge-by-edge (or at least vertex-by-vertex) basis, we must follow a more delicate approach.}

This discussion gives rise to \Cref{alg:rounding}, which we prove in this section provides the guarantees of \Cref{per-edge-guarantees}.

\begin{algorithm}[ht]
	\caption{}
	\label{alg:rounding}
	\begin{algorithmic}[1]
		\medskip
		\State \textbf{Init:} $\mathcal{M} \leftarrow \emptyset$ 
		\For{all vertices $v$, on arrival} 
		\State read $\{x_{u,v} \mid u<v \}$
		\If{$\sum_{u<v}x_{u,v} \leq \frac{1/2-c}{1/2+c}$}\label{line:low-deg-start}
		\State pick at most one $u<v$ with probability $x_{u,v}\cdot \frac{1/2+c}{1/2-c}$ \label{line:pick-scaled-up-for-low-deg-neighbor}
		\If{$u\neq \text{nil}$ and $u$ is unmatched in $\calM$}
		\With{\textbf{probability} $\frac{1/2-c}{g(u,v)}$}\label{line:scale-down-for-low-deg-neighbor}
		\State $\calM\gets \calM \cup \{(u,v)\}$ \label{line:low-deg-end}
		\EndWith
		\EndIf
		\Else\label{line:high-deg-start}
		\State pick at most one $u<v$ with probability $x_{u,v}$ \label{line:pick1}
		\If{$u\neq $ nil and $u$ is unmatched in $\mathcal{M}$}
		\With{\textbf{probability} $\min \left( 1, \frac{1/2 + c}{g(u,v)} \right)$} \label{line:probacceptfirstpick}
		\State $\mathcal{M} \leftarrow \mathcal{M} \cup \{ (u,v)\}$ \label{line:updateMfirstpick} \label{line:acceptfirstproposal} 
		\EndWith
		\EndIf
		\If{$v$ is still unmatched in $\mathcal{M}$}\label{line:second-pick-start}
		\State pick at most one $u<v$ with probability $x_{u,v}$ \label{line:pick2}
		\If{$u\neq $ nil and $u$ is unmatched in $\mathcal{M}$}
		\With{\textbf{probability} $p_{u,v}$ guaranteeing $\Pr[(u,v)\in \calM] = (1/2+c)\cdot x_{u,v}$} \label{line:probacceptsecondpick} 
		\State $\mathcal{M} \leftarrow \mathcal{M} \cup \{ (u,v)\}$ \label{line:acceptsecondproposal} \label{line:high-deg-end}
		\EndWith
		\EndIf
		\EndIf
		\EndIf
		\EndFor				
		\State \textbf{Output} $\mathcal{M}$
	\end{algorithmic}
\end{algorithm}


\subsection{High-Level Analysis}

For our analysis and proof of \Cref{per-edge-guarantees}, we will assume, by way of an inductive proof, that \Cref{invariant} holds for all edges $(u,w)$ with $u,w<v$ and therefore that for each $u<v$ we have $\Pr[F_{u,v}] = g(u,v)$, as stated in \Cref{prob-free}.

Given the inductive hypothesis, it is easy to verify that \Cref{alg:rounding} guarantees marginal probabilities of each edge to be matched to be precisely $(1/2+c)\cdot x_e$. 
Indeed, for an arriving vertex $v$ with low fractional degree, $\sum_{u<v}x_{u,v} \leq \frac{1/2-c}{1/2+c}$ (lines \ref{line:low-deg-start}-\ref{line:low-deg-end}), since by the inductive hypothesis $u$ is free at time $v$ with probability $\Pr[F_{u,v}] = g(u,v)$, we have that
$$\Pr[(u,v)\in \calM] = x_{u,v}\cdot \frac{1/2+c}{1/2-c} \cdot g(u,v)\cdot \frac{1/2-c}{g(u,v)} = (1/2+c)\cdot x_{u,v}.$$
In the alternative case of lines \ref{line:high-deg-start}-\ref{line:high-deg-end}, we trivially have that each edge $(u,v)$ with $u<v$ is matched with probability precisely $\Pr[(u,v)\in \calM] = (1/2+c)\cdot x_e$, due to lines \ref{line:probacceptsecondpick}-\ref{line:high-deg-end}.
The crux of the analysis, then, is in proving that this algorithm is well-defined, and in particular that there exists some probabilities $p_{u,v}$ as stated in Line \ref{line:probacceptsecondpick}.

We note that all probabilistic lines in the algorithm except for Line \ref{line:probacceptsecondpick} are trivially well-defined.
First, if $v$ has low fractional degree before time $v$, i.e., $\sum_{u<v}x_{u,v} \leq \frac{1/2-c}{1/2+c}$, then the probability of any neighbor to be picked in Line \ref{line:pick-scaled-up-for-low-deg-neighbor} is at most $\sum_{u<v} x_{u,v}\cdot \frac{1/2+c}{1/2-c} \leq 1$, and so this line is well-defined.
Next, by the fractional matching constraint, we have that $\sum_{u<v} x_{u,v}\leq 1$, and consequently lines \ref{line:pick1} and \ref{line:pick2} are well-defined.
Finally, by the fractional matching constraint, we have that $\sum_{w < v} (1/2+c)\cdot x_{u,w}\leq 1/2+c$, and therefore 
\begin{align}\label{prob-matched-bound}
\Pr[F_{u,v}] = g(u,v) \geq 1/2-c.
\end{align}
Consequently, the term $\frac{1/2-c}{g(u,v)}$ in Line \ref{line:scale-down-for-low-deg-neighbor} is indeed a probability, by our choice of $c=\bestc \leq 1/2$.
We now turn to proving that probabilities $p_{u,v}$ as stated in Line \ref{line:probacceptsecondpick} do indeed exist.

First, to show that $p_{u,v}\geq 0$, we must show that the probability of edge $(u,v)$ to be matched as a first pick in Line \ref{line:acceptfirstproposal} does not on its own exceed $(1/2+c)\cdot x_{u,v}$.
\begin{obs}\label{obs:first-pick-UB}
	The probability of an edge $(u,v)$ to be matched in Line \ref{line:acceptfirstproposal} is at most
	$$\Pr[(u,v) \textrm{ added to $\calM$ in Line \ref{line:acceptfirstproposal}}] \leq (1/2+c)\cdot x_{u,v}.$$
\end{obs}
\begin{proof}
	By the inductive hypothesis, we have that $\Pr[F_{u,v}] = g(u,v)$. Consequently, 
	\begin{align*}
	\Pr[(u,v)\textrm{ added to $\calM$ in Line \ref{line:acceptfirstproposal}}] &= x_{u,v} \cdot \min\left(1,\frac{1/2+c}{g(u,v)}\right)\cdot g(u,v) \leq (1/2+c)\cdot x_{u,v}.\qedhere
	\end{align*}
\end{proof}
\begin{cor}\label{puv>=0}
	The parameter $p_{u,v}$ in Line \ref{line:probacceptsecondpick} satisfies $p_{u,v}\geq 0$.
\end{cor}

The core of the analysis will then be in proving that $p_{u,v}\leq 1$. For this, we will need to argue that a second pick in lines \ref{line:second-pick-start}-\ref{line:high-deg-end} is likely to result in $(u,v)$ being matched, provided we set $p_{u,v}\leq 1$ high enough.  We prove as much in the next section.

\subsection{Core of the Analysis}\label{sec:core}

In this section we prove that the second pick is likely to result in a match. 
To this end, we prove that the matched statuses of neighbors of an arriving vertex $v$ have low positive correlation (if any).
More formally, if $G=(V_1,V_2,E)$ is our bipartite graph, 
we will prove the following.

\begin{lem}\label{bounded-covariance}
	For any $i=1,2$, vertex $v$ and vertices $u,w<v$ with $u,w\in V_i$, 
	$$\Cov(F_{u,v}, F_{w,v})\leq 6c.$$
\end{lem}

Since the covariance of two binary variables $A$ and $B$ is equal to that of their complements, $\Cov(A,B)=\Cov(1-A, 1-B)$, we will concern ourselves with bounding $\Cov(M_{u,v}, M_{w,v})$, where $M_{u,v}:=1-F_{u,v}$ is an indicator for $u$ being matched in $\calM$ before $v$ arrives.

For this proof, we write $M_{u,v}$ as the sum of two Bernoulli variables, $M_{u,v} = M^L_{u,v} + M^H_{u,v}$. 
The indicators $M^L_{u,v}$ and $M^H_{u,v}$ correspond to $u$ being matched to some neighbor $w$ at a time $z$ when $u$ had low or high fractional degree, respectively. That is, 
\begin{align*}
M^L_{u,v} := \mathds{I}\left[(u,w)\in \calM \textrm{ for some } w<v \textrm{ with } \sum_{z<\min\{u,w\}} x_{u,z} \leq \frac{1/2-c}{1/2+c}\right],
\end{align*}
with $M^H_{u,v} = M_{u,v} - M^L_{u,v}$ defined analogously.

In what follows, we will show that for any vertex $v$ and index $i=1,2$, the variables $\{M^L_{u,v} \mid u\in V_{i}\}$ are negatively correlated, while the variables $\{M^H_{u,v} \mid u\in V_{i}\}$ have low probability, which implies that they have low positive correlation with any other binary variable. These bounds will allow us to bound the correlation of the sums $M_{u,v} = M^L_{u,v} + M^H_{u,v}$.

We start by proving the negative correlation between $M^L_{u,v}$ variables, and indeed proving negative association of these variables. 


\begin{lem}\label{matching-NA}
	For any $i=1,2$ and vertex $v$, the variables $\{M^L_{u,v} \mid u<v,\,\, u\in V_{i}\}$ are NA.
\end{lem}
By \Cref{NA:neg-corr}, this implies that the above variables are negatively correlated.
\begin{cor}\label{neg-cor-M1}
	For any $i=1,2$, vertex $v$ and earlier vertices $u,w<v$ with $u,w\in V_i$, 
	$$\Cov(M^L_{u,v}, M^L_{w,v})\leq 0.
	$$
\end{cor}

\begin{proof}[Proof of \Cref{matching-NA}]
	Recall that $M^L_{u,v}$ is an indicator for $u$ being matched before arrival of $v$ before it has high fractional degree.
	By definition of \Cref{alg:rounding}, this implies that a matching event accounted for by $M^L_{u,v}$ can only occur in lines \ref{line:low-deg-end} or \ref{line:acceptfirstproposal}.
    Such matches occur due to $u$ picking a neighbor or being picked as a neighbor in line \ref{line:pick-scaled-up-for-low-deg-neighbor} or \ref{line:pick1}, and the probabilistic test in line \ref{line:scale-down-for-low-deg-neighbor} or \ref{line:probacceptfirstpick} (respectively), passing, if the picked vertex was previously unmatched in $\calM$.
    We imagine we perform the probabilistic tests in lines \ref{line:scale-down-for-low-deg-neighbor} and \ref{line:probacceptfirstpick} \emph{before} testing whether the picked vertex was unmatched in $\calM$.
    
    For vertices $w<z$, let $A_{w,z}$ be an indicator for $z$ picking $w$ in line \ref{line:pick-scaled-up-for-low-deg-neighbor} or \ref{line:pick1}, and the probabilistic test in line \ref{line:scale-down-for-low-deg-neighbor} or \ref{line:probacceptfirstpick} (respectively) passing.
	Then, by the 0-1 Principle (\Cref{0-1-NA}), we have that for any vertex $z$, the variables $\{A_{w,z} \mid w<z\}$ are NA.
	Moreover, the families of variables $\{A_{w,z} \mid w<z\}$ for distinct $z$ are NA. Therefore, by closure of NA under independent union (\Cref{NA:ind-union}), the variables $\{A_{w,z} \mid z,\, w<z\}$ are NA.
	For notational simplicity, letting $A_{z,w}:=A_{w,z}$ for $z>w$ (recall that we only defined $A_{w,z}$ for $w<z$), we find that if $z'$ is the smaller of $v-1$ and the first time $z$ that $u$ has high fractional degree, the variables $M^L_{u,v}$ are precisely equal to 
	\begin{align*}
	M^L_{u,v} := \bigvee_{w \leq z'} A_{w,u}.
	\end{align*}
	Indeed, this is due to $u$ being matched while it has low fractional degree upon the first time that it is picked by a neighbor (or it picks a neighbor) in line \ref{line:pick-scaled-up-for-low-deg-neighbor} or \ref{line:pick1}, and the corresponding probabilistic test in line \ref{line:scale-down-for-low-deg-neighbor} or \ref{line:probacceptfirstpick} passes.
	Therefore, by closure of NA under monotone function composition (\Cref{NA:fn-comp}), the variables $\{M^L_{u,v} \mid u\in V_i\}$, which are monotone nondecreasing functions of disjoint subsets of the variables $A_{w,u}$ by bipartiteness, are NA.\footnote{This is the only place in our analysis where we use bipartiteness.}
\end{proof}

We now turn to upper bounding the probability of the event $M^H_{u,v}$.

\begin{lem}\label{low-prob-M2}
	For any edge $(u,v)$ with $u<v$, we have that 
	$\Pr[M^H_{u,v}] \leq 2c.$
\end{lem}
\begin{proof}
    Recall that by the inductive hypothesis, $\Pr[(u,w)\in \calM]=(1/2+c)\cdot x_{u,w}$. On the other hand, by the fractional matching constraint, we have that $\sum_{w<v} x_{u,v} \leq 1$, and therefore $\Pr[M_{u,v}]\leq 1/2+c$. On the other hand, if we denote by $z_u$ the first time $u$ has high fractional degree, then either $z_u\geq v$, in which case $\Pr[M^H_{u,v}]=0$, or
    \begin{align*}\Pr[M^L_{u,v}] \geq \sum_{w<z_u} x_{u,w}\cdot (1/2+c) \geq \frac{1/2-c}{1/2+c}\cdot (1/2+c) = 1/2-c,
    \end{align*}
    in which case we have
    \begin{align*}
        \Pr[M^H_{u,v}] & = \Pr[M_{u,v}] - \Pr[M^L_{u,v}] \leq 2c. \qedhere
    \end{align*}
\end{proof}

We are now ready to prove \Cref{bounded-covariance}, whereby vertices $u,w$ on the same side of the bipartition have weakly correlated matched statuses, namely $\Cov(F_{u,v}, F_{w,v}) \leq 6c$.
\begin{proof}
	By definition of covariance, the binary variables $F_{u,v}$ and $F_{w,v}$ satisfy $\Cov(F_{u,v}, F_{w,v}) = \Cov(1-F_{u,v}, 1-F_{w,v}) = \Cov(M_{u,v}, M_{w,v})$ (see \Cref{covariance-of-complements}). We therefore turn to upper bounding the covariance of the variables $M_{u,v}$ and $M_{w,v}$.
	
	By the additive law of covariance, the covariance of the variables $M_{u,v} = M^L_{u,v} + M^H_{u,v}$ and $M_{w,v} = M^L_{w,v} + M^H_{w,v}$, denoted by $(\star) = \Cov(M_{u,v}, M_{w,v})$, satisfies
	\begin{align*}
	(\star) & = \Cov(M^L_{u,v} + M^H_{u,v}\,,\, M^L_{w,v} + M^H_{w,v}) \\
	& = \Cov(M^L_{u,v}, M^L_{w,v}) + \Cov(M^L_{u,v}, M^H_{w,v}) + \Cov(M^H_{u,v}, M^L_{w,v}) + \Cov(M^H_{u,v}, M^H_{w,v}) \\
	& \leq 0 + \Pr[M^L_{u,v}, M^H_{w,v}] + \Pr[M^H_{u,v}, M^L_{w,v}] + \Pr[M^H_{u,v}, M^H_{w,v}]  \\
	& \leq 0 + \Pr[M^H_{w,v}] + \Pr[M^H_{u,v}] + \Pr[M^H_{u,v}] \\
	& \leq 6c.
	\end{align*}
	Here, the first inequality follows from \Cref{neg-cor-M1}, the second inequality follows from the trivial bound on covariance of Bernoulli variables $A$ and $B$ given by $\Cov(A,B) = \Pr[A,B] - \Pr[A]\cdot \Pr[B] \leq \Pr[A,B]\leq \Pr[A]$, and the final inequality follows from \Cref{low-prob-M2}.
\end{proof}

\Cref{bounded-covariance} now allows us to argue that if $u$ has high degree upon arrival of $v$, 
then $F_{u,v}$ is nearly independent of the event $R_v$, whereby $v$ is rejected (not matched) after its first pick of $u_1$ (possibly $u_1= \textrm{nil}$). In particular, we have the following.
\begin{lem}\label{free-reject-pick-w-prob}
	Let $u<v$ be a vertex of high fractional degree, $\sum_{w<v} x_{u,w}$, upon arrival of $v$. Then, for all $w\neq u$ (including possibly $w=\text{nil}$), we have 
	\begin{align*}
	\Pr[F_{u,v}, R_v, u_1=w] \geq \Pr[F_{u,v}]\cdot \Pr[R_v, u_1 = w]\cdot \left(1-\frac{6c}{(1/2-c)^2}\right).
	\end{align*}
\end{lem}
\begin{proof}
	For $w=\text{nil}$ the claim follows from the event $u_1=\text{nil}$ implying $R_v$, and being independent of $F_{u,v}$.
	\begin{align*}
	\Pr[F_{u,v}, R_v, u_1=\text{nil}] & = \Pr[F_{u,v}]\cdot \Pr[ R_v, u_1=\text{nil}].
	\end{align*}
	
	Next, let $w<v$ be some neighbor of $v$. If we denote by $q_{w,v} := \min\left(1,\frac{1/2+c}{g(w,v)}\right)$ the probability that $w$ does not reject $v$ if it is picked first and is free, then the probability that $u_1=w$ and $v$ gets rejected in its first pick is 
	\begin{align}\label{prob-Rv-u1=w}
	\Pr[R_v, u_1=w] = x_{w,v}\cdot \left(1- q_{w,v} \cdot \Pr[F_{w,v}]\right)\geq x_{w,v}\cdot (1/2-c),
	\end{align}
	where the inequality follows from $\Pr[F_{w,v}] = g(w,v)$ by \Cref{prob-matched-bound}, which implies that $q_{w,v} \cdot \Pr[F_{w,v}] \leq 1/2+c$.
	Similarly, the probability $u$ is free, $u_1=w$ and $v$ gets rejected in its first pick is 
	\begin{align}
	\Pr[F_{u,v}, R_v, u_1=w] & = x_{w,v}\cdot \left(\Pr[F_{u,v}]-q_{w,v}\cdot \Pr[F_{w,v}, F_{u,v}]\right) \nonumber \\
	& \geq x_{w,v}\cdot \left(\Pr[F_{u,v}]-q_{w,v} \cdot (\Pr[F_{w,v}]\cdot \Pr[F_{u,v}] + 6c)\right), \nonumber 
	\\
	& \geq x_{w,v}\cdot \left(\Pr[F_{u,v}]-q_{w,v} \cdot (\Pr[F_{w,v}]\cdot \Pr[F_{u,v}]) - 6c\right) \nonumber \\
	& \geq x_{w,v}\cdot \Pr[F_{u,v}]\cdot \left(1-q_{wv}\cdot \Pr[F_{w,v}] - \frac{6c}{1/2-c} \right) \nonumber \\
	& \geq \Pr[F_{u,v}]\cdot \Pr[R_v, u_1 = w]\cdot \left(1-\frac{6c}{(1/2-c)^2}\right), \nonumber
	\end{align}
	where the first inequality follows from \Cref{bounded-covariance}, the second inequality follows from the trivial bound $q_{wv}\leq 1$, 
	the third inequality follows from $\Pr[F_{u,v}] = g(u,v)\geq 1/2-c$ by \Cref{prob-free}, and the final inequality follows from \Cref{prob-Rv-u1=w}.
\end{proof}

In what follows we denote by $x_{\text{nil},v} := 1-\sum_{w<v} x_{w,v}$ the probability with which $u_1 = \text{nil}$. 
From \Cref{free-reject-pick-w-prob} and \Cref{prob-Rv-u1=w}, as well as $\Pr[R_v, u_1=\text{nil}] = \Pr[u_1 = \text{nil}] = x_{\text{nil}, v}$, 
we obtain the following lower bound on $\Pr[F_{u,v}, R_v, u_1=w]$ in terms of $x_{w,v}$.
\begin{cor}\label{free-reject-pick-w-prob-as-fn-of-x} 
	For any vertex $v$ and $w$ (possibly $w=\text{nil}$), we have that $$\Pr[F_{u,v}, R_v, u_1 = w] \geq \Pr[F_{u,v}]\cdot x_{w,v}\cdot \left(1/2-c-\frac{6c}{1/2-c}\right).$$
\end{cor}

Finally, we are ready to prove that $p_{u,v}$ is a probability, and in particular $p_{u,v}\leq 1$.
\begin{lem}
	The parameter $p_{u,v}$ in Line \ref{line:probacceptsecondpick} satisfies $p_{u,v}\in [0,1]$.
\end{lem}
\begin{proof}
	Non-negativity of $p_{u,v}$ was proven in \Cref{puv>=0}. We turn to proving that $p_{u,v}\leq 1$ suffices to guarantee $\Pr[(u,v)\in \calM]\geq (1/2+c)\cdot x_{u,v}$, from which we obtain that there exists some $p_{u,v}\in [0,1]$ which results in $\Pr[(u,v)\in \calM] = (1/2+c)\cdot x_{u,v}$.
	
	By \Cref{prob-matched-bound} we have that $\Pr[F_{u,v}] = g(u,v)\geq 1/2-c$, and therefore 
	\begin{align}\label{M1-match-lb}
	\Pr[(u,v)\in \calM \textrm{ in Line \ref{line:acceptfirstproposal}}] = x_{u,v}\cdot \min\left(1,\frac{1/2+c}{g(u,v)}\right)\cdot g(u,v) \geq (1/2-c)\cdot x_{u,v}.
	\end{align}
	We therefore wish to prove that the probability of $(u,v)$ being matched in Line \ref{line:acceptsecondproposal} is at least $2c\cdot x_{u,v}$, for some choice of $p_{u,v}\leq 1$.
	And indeed, 
	\begin{align*}
	\Pr[(u,v)\textrm{ added to $\calM$ in Line \ref{line:acceptsecondproposal}}] 
	& = x_{u,v} \cdot \sum_{w\neq u} \Pr[F_{u,v}, R_v, u_1=w]\cdot p_{u,v} \\
	& \geq x_{u,v}\cdot \Pr[F_{u,v}] \cdot \sum_{w\neq u}  x_{w,v}\cdot \left(1/2-c-\frac{6c}{1/2-c} \right)\cdot p_{u,v} \\
	& \geq x_{u,v}\cdot (1/2-c) \cdot (1-4c)\cdot \left(1/2-c-\frac{6c}{1/2-c} \right)\cdot p_{u,v} \\
	& \geq 2c\cdot x_{u,v},
	\end{align*}
	where the first inequality follows from \Cref{free-reject-pick-w-prob-as-fn-of-x} and \Cref{properties:c}. The second inequality holds due to \Cref{prob-free} implying $\Pr[F_{u,v}] \geq 1/2-c$ and due to vertex $u$ having high degree at time $v$, and therefore by the fractional matching constraint $x_{u,v} \leq 1-\frac{1/2-c}{1/2+c} = \frac{2c}{1/2+c}\leq 4c$, and hence $\sum_{w\neq u} x_{w,v}\geq 1-4c\geq 0$ (again using \Cref{properties:c}). The final inequality holds 
	for $p_{u,v} = 1$ and for our choice of $c$, by \Cref{def:c}.
	
	
	Consequently, combining the above with \Cref{M1-match-lb}, we find that setting $p_{u,v}=1$ results in $(u,v)$ being matched in either Line \ref{line:acceptfirstproposal} or Line \ref{line:acceptsecondproposal} with probability at least 
	\begin{align}\label{prob-for-puv=1}
	\Pr[(u,v)\in \calM] & \geq (1/2+c)\cdot x_{u,v}.
	\end{align}
	As the probability of $(u,v)$ being added to $\calM$ in Line \ref{line:acceptsecondproposal} is monotone increasing in $p_{u,v}$, we conclude that there exists some $p_{u,v}\in [0,1]$ for which \Cref{prob-for-puv=1} holds with equality.
\end{proof}

\noindent\textbf{Conclusion of \Cref{alg:rounding}'s analysis.}
To conclude, \Cref{alg:rounding} is well-defined, and this algorithm outputs a random matching $\calM$ which matches each edge $e$ with probability precisely $\Pr[e\in \calM] = (1/2+c)\cdot x_e$. \Cref{per-edge-guarantees} follows.

\begin{remark}\textbf{Computational Aspects:} 
    As described, the only way we are aware of to implement \Cref{line:probacceptsecondpick} exactly (and in particular, computing all $p_{u,v}$ exactly) is using an exponential-time algorithm maintaining the joint distributions as they evolve. 
	However, a simple modification of the algorithm, resulting in a polynomial-time algorithm with a $(1+o(1))$ additional multiplicative loss in each edge's matching probability, can be readily obtained by approximately estimating the above $p_{u,v}$ up to $(1\pm o(1))$ multiplicative errors, by standard monte carlo methods. 
	As this results in rather cumbersome descriptions and subsequent calculations, and since running time is not our focus, we do not expand on this.
\end{remark}
\section{Putting it all Together}\label{sec:wrap-up}

In this section we prove our main result, \Cref{thm:greedy-suboptimal}, restated below for ease of reference.
\greedysubopt*

\begin{proof}
	For a graph of maximum degree at most $\Delta$, assigning $x$-value $1/\Delta$ to each edge yields a fractional matching.
	Applying \Cref{alg:rounding} to this fractional matching in a bipartite graph under vertex arrivals results in each edge being matched with probability $\halfplusbestc/\Delta$, by \Cref{per-edge-guarantees}.
	Therefore, by \Cref{coloring-to-matching}, there exists an online edge coloring algorithm whose competitive ratio is  $(1/\halfplusbestc+o(1)) \approx \approxratio+o(1)$ w.h.p.~on bipartite graphs of maximum degree $\Delta = \omega(\log n)$ under (interleaved) vertex arrivals.
	Finally, \Cref{random-subgraphs} together with union bound implies that the same competitive ratio (up to $o(1)$ terms) carries over to general graphs under vertex arrivals.
\end{proof}

\begin{remark}
	Our analysis extends to prove the slightly tighter result, whereby there exist constants $c_1,c_2>0$ and a $(2-c_1)$-competitive online algorithm for $n$-node graphs of maximum degree at least $c_2\cdot \log n$ under vertex arrivals. (See \Cref{remark1}.) For brevity's sake, we omit the details.
\end{remark}
\section{Conclusion}
In this work we resolve the longstanding conjecture of Bar-Noy, Motwani and Naor, namely \Cref{conj:bmn}. That is, we show that, while for bounded-degree graphs the greedy algorithm's competitive ratio of $2$ is optimal among online algorithms, for high-degree graphs this is not the case. 

Some natural questions remain. What is the best achievable competitive ratio? Is a ratio of $1+o(1)$ possible, as for one-sided arrivals in bipartite graphs and random-order edge arrivals \cite{cohen2019tight,bhattacharya2021online}? Can the same be achieved under adversarial \emph{edge} arrivals? Bar-Noy et al.~\cite{bar1992greedy} suggested a candidate algorithm for this latter model, but its analysis seems challenging.
Finally, does the online rounding \Cref{alg:rounding} have more applications beyond edge coloring?

\paragraph{Acknowledgements}
We thank Janardhan Kulkarni for drawing our attention to \cite{karloff1987efficient}.
This research is partially supported by NSF award 1812919, ONR award N000141912550, and a gift from Cisco Research.

\section*{Appendix}
\appendix
\section{The Karloff-Shmoys Approach: Online}\label{sec:karloff-shmoys}

Here we substantiate our earlier assertion that $\alpha$-competitive online edge coloring on high-degree graphs is equivalent (up to $o(1)$ terms) to the same task on high-degree \emph{bipartite} graphs.
That is, we outline the proof of \Cref{random-subgraphs}, restated below for ease of reference.

\karloffshmoys*

\begin{proof}
	The general graph edge coloring algorithm relies on the following subroutine for sampling balanced random subgraphs in subgraphs of maximum degree $\Delta' \geq 18\cdot \sqrt{\Delta \log n}$. (Note that $\Delta\geq 18\sqrt{\Delta \log n}$, by the hypothesis, whereby $\Delta = \omega(\log n)$.)
	Assign each vertex to a set $V_i\subseteq V$ with $i=1,2$ chosen uniformly at random. For any vertex $v\in V$, let $d(v)$ denotes the degree of $v$ in $G$, and $D_v$ denotes the (random) degree of $v$ in the random bipartite subgraph $H=H(V_1,V_2,E\cap (V_1\times V_2))$. Then, we have that $\E[D_v] = d(v)/2\leq \Delta'/2$. 
	By Chernoff's Bound (\Cref{chernoff}), since $D_v$ is the sum of independent $\textrm{Bernoulli}(1/2)$ variables, we have that, for $\eps=\sqrt[4]{\log n/\Delta} = o(1)$, 
	\begin{align}\label{degree-bound-recursive-coloring}
	\Pr[D_v \geq (\Delta'/2)\cdot (1+\eps)]\leq \exp\left(\frac{-(\Delta'/2)\cdot \eps^2}{3}\right) \leq \frac{1}{n^3},
	\end{align}
	using $\Delta' \geq 18\cdot \sqrt{\Delta \cdot \log n}$, and consequently $\Delta\cdot \eps^2 \geq 18\log n$.
	The same high-probability bound holds for $d(v)-D_v$, which is identically distributed to $D_v$.

	To achieve an online edge coloring algorithm for $G$ from the above, we apply the $\alpha$-competitive edge coloring algorithm to the random bipartite $H$, and recursively apply the same approach to the random subgraph induced by the edges outside of $H$, namely $G\setminus H = G[E\setminus (V_1\times V_2)]$, until $H$ is guaranteed to have degree at most $18\cdot \sqrt{\Delta\cdot \log n}$ w.h.p.
	We note that this approach can be applied online, by assigning to each vertex $v$ on arrival a side of each of the recursive random bipartitions.
	Moreover, the colors of each recursive level number $\ell$ can be associated with a contiguous set of integers of cardinality $\alpha\cdot \Delta\cdot ((1+\eps)/2)^\ell$, which is the high probability upper bound on the number of colors used in this recursive call.
	Repeating the above recursively for $t:=\log_{2/(1+\eps)} (18 \sqrt{\Delta/\cdot \log n}) \leq \log n$ levels results in a random uncolored subgraph of maximum degree at most $18\sqrt{\Delta \cdot \log n} = o(\Delta)$ w.h.p., which we color greedily.
	
	Taking union bound over the $O(n^2)$ bad events (some vertex degree $D_v$ exceeding $\Delta'\cdot ((1+\eps)/2)$ in a random bipartite subgraph or its complement in a subgraph whose maximum degree is $\Delta'\geq 18\sqrt{\Delta\cdot \log n}$, or any of the bipartite edge coloring algorithms failing to be $\alpha$ competitive on the subgraph it is applied to), we have that w.h.p., the number of colors $C$ used is, as desired, at most
	\begin{align*}
		C & \leq \alpha\cdot \Delta\cdot \frac{1+\eps}{2} + \alpha\cdot \Delta\cdot \left(\frac{1+\eps}{2}\right)^2 
		+ \dots + \alpha\cdot \Delta\cdot \left(\frac{1+\eps}{2}\right)^t + 
		36\cdot \sqrt{\Delta \cdot \log n} \\
		& \leq \alpha\cdot \sum_{i\geq 1} \Delta\cdot \left(\frac{1+\eps}{2}\right)^i + 36\cdot \sqrt{\Delta \cdot \log n} \\
		& = \alpha\cdot \Delta \cdot \frac{1+\eps}{1-\eps} + o(\Delta) \\
		& = (\alpha+o(1))\cdot \Delta. \qedhere
	\end{align*}
\end{proof}

\begin{remark}
	As stated in the introduction, we note that the above reduction from general to bipartite graphs results in bipartite graphs with \emph{interleaved} vertex arrivals.
\end{remark}
\section{Edge Coloring from Random Matchings}\label{sec:reducing-coloring-to-matching}

In this section, we show how to reduce edge coloring in (bipartite) graphs under vertex arrivals to computing a random matching which matches each edge with probability $\Omega(1/\Delta)$.

\coloringtomatching*
\begin{proof}
	If $\alpha>2$, then the claim follows trivially from the greedy algorithm's $2$-competitiveness. We therefore assume $\alpha\leq 2$. We give a subroutine which decreases the uncolored degree of a subgraph of maximum degree $\Delta' \geq 48\cdot \sqrt[4]{\Delta^3 \log n}$ at a rate of one per $\alpha+o(1)$ colors w.h.p. (Note that $\Delta\geq 48\sqrt[4]{\Delta^3 \log n}$, by the hypothesis, whereby $\Delta = \omega(\log n)$.)
	
	Our subroutine is as follows. Let $L:=12\sqrt{\Delta \log n}$ and $\eps:=\sqrt[4]{(\log n)/\Delta} (= o(1) \leq 1/2)$. 
	We note that by our choice of $L$ and $\eps$ and our lower bound on $\Delta'$, we have that 
	\begin{align}\label{L/D}
	4L/\Delta' \leq 48\sqrt{\Delta \log n}/48\sqrt[4]{\Delta^3 \log n} = \sqrt[4]{(\log n)/\Delta} = \eps.
	\end{align}
	For $i=1,\dots,\lceil \alpha \cdot L\rceil$, we run Algorithm $\calA$, which matches each edge with probability at least $(1/\alpha)/\Delta'$, and color all previously-uncolored matched edges in this run of $\calA$ using a new (common) color. 
	Fix a vertex $v$ whose degree in the subgraph is at least $d(v)\geq \Delta' - \lceil \alpha \cdot L\rceil$ and let $X_1,\dots, X_L$ be indicators of $v$ having an edge colored during application $i=1,\dots,\lceil \alpha\cdot L\rceil $ of Algorithm $\calA$.
	Since vertex $v$ can have at most $\lceil \alpha\cdot L\rceil \leq 2\cdot L$ edges colored during these $L$ applications of Algorithm $\calA$, we find that the number of uncolored edges of $v$ at any point during this subroutine is at least $\Delta' - 2\lceil \alpha \cdot L\rceil \geq \Delta' - 4L$, independently of previous random choices. 
	On the other hand, since each uncolored edge is matched (and hence colored) with probability at least $(1/\alpha)/\Delta'$, we have that for any history $\mathcal{H}$ of random choices in applications $1,2,\dots,i-1$ of $\calA$, application $i$ of $\calA$ results in one of the (at least) $\Delta'-4L$ uncolored edges of $v$ being colored with probability at least
	\begin{align}\label{conditional-color-prob}
		\Pr[X_i \mid \mathcal{H}] \geq (1/\alpha)\cdot (\Delta' - 4L)/\Delta' = (1/\alpha)\cdot (1-4L/\Delta') \geq (1/\alpha)\cdot (1-\eps),
	\end{align}	
	where the last inequality relied on \Cref{L/D}.
	Combining \Cref{conditional-color-prob} with standard coupling arguments (\Cref{coupling}) together with a Chernoff Bound (\Cref{chernoff}), we find that the number of colored edges of $v$, denoted by $X:=\sum_i X_i$ satisfies
	\begin{align*}
		\Pr[X \leq L \cdot (1-\eps)^2] & \leq \exp\left(\frac{-L\cdot (1-\eps)\cdot \eps^2}{2}\right) \leq \exp\left(\frac{-L \cdot \eps^2}{4}\right) = \frac{1}{n^3},
	\end{align*}
	where the second inequality follows from $\eps\leq 1/2$ and the equality follows from choice of $L$ and $\eps$.
	Union bounding over the $n$ vertices, we obtain the following high probability bound on the maximum degree of the uncolored subgraph $H$ after the $\lceil \alpha \cdot L\rceil$ applications of $\calA$:
	\begin{align}\label{subroutine-deg-bound}
	\Pr[\Delta(H) \geq \Delta' - L\cdot (1-\eps)^2] \leq \frac{1}{n^2}.
	\end{align}
	
	We now describe how to make use of this subroutine. 
	For $r=1,\dots,\Delta/L$ \emph{phases}, let $\Delta_i := \Delta - (i-1)\cdot L\cdot (1-\eps)^2.$
	If $\Delta_i < 48\sqrt[4]{\Delta^3 \log n}$, apply the greedy coloring. Otherwise, apply the above subroutine with $\Delta'=\Delta_i$.
	A simple inductive argument together with union bound, relying on \Cref{subroutine-deg-bound}, shows that for $i=1,2,\dots,\Delta/L(\leq n)$, the uncolored subgraph after the first $i-1$ phases has maximum degree at most $\Delta'\leq \Delta_i$ w.h.p., or alternatively it has maximum degree at most $\Delta' \leq 48\cdot \sqrt[4]{\Delta^3 \log n} = o(\Delta)$.
	Moreover, each of these $\Delta/L$ phases requires at most $\lceil \alpha\cdot L\rceil\leq \alpha\cdot L + 1$ colors, by definition, and therefore these $\Delta/L$ phases require at most $\alpha\cdot \Delta + \Delta/L = (\alpha+o(1))\cdot \Delta$ colors in total. Finally, after these phases we are guaranteed that the maximum degree of the uncolored subgraph is at most $\min\{ 48\cdot \sqrt[4]{\Delta^3 \log n}, \, \Delta - (\Delta/L)\cdot L\cdot (1-\eps)^2\} = o(\Delta)$. Applying the greedy algorithm to this uncolored subgraph after the $\Delta/L$ phases thus requires a further $2\cdot o(\Delta) =o(\Delta)$ colors. This results in a proper edge coloring using $(\alpha+o(1))\cdot \Delta$ colors w.h.p.
	
	Finally, we note that the above algorithm can be implemented online under vertex arrivals, since $\calA$ works under vertex arrivals. 
	In particular, when a vertex arrives, we perform the next steps of the different copies of Algorithm $\calA$ (with the different settings of $\Delta_i$) on the uncolored subgraphs obtained from each phase, simulating the arrival of a vertex in each such uncolored subgraph.
	Combined with the above, this yields the desired result: an edge coloring algorithm which is $(\alpha+o(1))$-competitive on general $n$-node graphs of maximum degree $\Delta = \omega(\log n)$ under vertex arrivals.
\end{proof}
\begin{remark}
	\Cref{coloring-to-matching} naturally extends to edge arrivals. Unfortunately, no algorithm matching each edge with probability $(1/\alpha)/\Delta$ subject to edge arrivals is currently known for any constant $\alpha<2$.
\end{remark}

\begin{remark}
	The approach of \Cref{coloring-to-matching} only requires matching algorithms which match each edge with probability $(1/\alpha)/\Delta$ for \emph{subgraphs} of the input graph. 
	Consequently, improved matching algorithms, with smaller $\alpha\geq 1$, for any downward-closed family of graphs $\mathcal{F}$ imply a similar improved $(\alpha+o(1))$-competitive edge coloring algorithm for the same family.
\end{remark}
\subsection{Probability Basics}
Here we include, for completeness, a number of basic probabilistic results used in this paper.

\begin{prop}[Chernoff Bound]\label{chernoff}
	Let $X=\sum_i X_i$ be the sum of independent Bernoulli random variables $X_i\sim \text{Bernoulli}(p_i)$, with expectation $\mu:=\E[X]=\sum_i p_i$.
	Then, for any $\eps\in (0,1)$, and $\kappa\geq \mu$,
	$$\Pr[X \geq \kappa\cdot (1+\eps)] \leq \exp\left(\frac{-\kappa \cdot \eps^2}{3}\right).$$
	$$\Pr[X \leq \mu\cdot (1-\eps)] \leq \exp\left(\frac{-\mu\cdot \eps^2}{2}\right).$$	
\end{prop}

\begin{prop}[Coupling]\label{coupling}
	Let $X_1,\dots,X_m$ be binary random variables such that for all $i$ and $\vec{x}\in \{0,1\}^{i-1}$, 
	$$\Pr\left[X_i = 1 \,\Bigg\vert\, \bigwedge_{\ell\in [i-1]} (X_\ell = x_\ell)\right] \geq p_i.$$
	If $\{Y_i\sim Bernoulli(p_i)\}_i$ are independent random variables, then for any $k\in \mathbb{R}$, 
	$$\Pr\left[\sum_i X_i \leq k\right] \leq \Pr\left[\sum_i Y_i \leq k\right].$$
\end{prop}

\begin{prop}\label{covariance-of-complements}
	Let $A$ and $B$ be Bernoulli random variables. Then $$\Cov(A,B)=\Cov(1-A, 1-B).$$
\end{prop}

\bibliographystyle{plainurl}%
\bibliography{abb,ultimate}

\end{document}